\documentclass[envcountsect]{llncs}
\usepackage{amsmath,amssymb,newlfont,graphicx,caption,verbatim}
\usepackage[ruled,lined,boxed,commentsnumbered,linesnumbered]{algorithm2e}
\newcommand{\dom}{\mathrm{dom}}
\newtheorem{notation}[theorem]{Notation}
\newcommand{\len}{\mathit{len}}
\newcommand{\poly}{\mathsf{poly}}

\title{$\delta$-Complete Decision Procedures for\\Satisfiability~over~the Reals\thanks{This research was sponsored by the National Science Foundation grants no. DMS1068829, no. CNS0926181 and no. CNS0931985, the GSRC under contract no. 1041377 (Princeton University), the Semiconductor Research Corporation under contract no.~2005TJ1366, General Motors under contract no.~GMCMUCRLNV301, and the Office of Naval Research under award no.~N000141010188.}
}
\author{Sicun Gao \and Jeremy Avigad \and Edmund M. Clarke}
\institute{Carnegie Mellon University, Pittsburgh, PA 15213}

\begin{document}
\maketitle

\begin{abstract}
We introduce the notion of ``$\delta$-complete decision procedures'' for solving SMT problems over the real numbers, with the aim of handling a wide range of nonlinear functions including transcendental functions and solutions of Lipschitz-continuous ODEs. Given an SMT problem $\varphi$ and a positive rational number $\delta$, a $\delta$-complete decision procedure determines either that $\varphi$ is unsatisfiable, or that the ``$\delta$-weakening'' of $\varphi$ is satisfiable. Here, the $\delta$-weakening of $\varphi$ is a variant of $\varphi$ that allows $\delta$-bounded numerical perturbations on $\varphi$. We prove the existence of $\delta$-complete decision procedures for bounded SMT over reals with functions mentioned above. For functions in Type 2 complexity class $\mathsf{C}$, under mild assumptions, the bounded $\delta$-SMT problem is in $\mathsf{NP^C}$. This stands in sharp contrast to the well-known undecidability results. $\delta$-Complete decision procedures can exploit scalable numerical methods for handling nonlinearity, and we propose to use this notion as an ideal requirement for numerically-driven decision procedures. As a concrete example, we formally analyze the DPLL$\langle$ICP$\rangle$ framework, which integrates Interval Constraint Propagation (ICP) in DPLL(T), and establish necessary and sufficient conditions for its $\delta$-completeness. We discuss practical applications of $\delta$-complete decision procedures for correctness-critical applications including formal verification and theorem proving.\end{abstract}

\section{Introduction}

Given a first-order signature $\mathcal{L}$ and a structure $\mathcal{M}$, the {\em Satisfiability Modulo Theories} (SMT) problem asks whether a quantifier-free $\mathcal{L}$-formula is satisfiable over $\mathcal{M}$, or equivalently, whether an existential $\mathcal{L}$-sentence is true in $\mathcal{M}$. Solvers for SMT problems have become the key enabling technology in formal verification and related areas. SMT problems over the real numbers are of particular interest, because of their importance in verification and design of hybrid systems, as well as in theorem proving. While efficient algorithms~\cite{linear06} exist for deciding SMT problems with only linear real arithmetic, practical problems normally contain nonlinear polynomials, transcendental functions, and differential equations. Solving formulas with these functions is inherently intractable. Decision algorithms~\cite{collins} for formulas with nonlinear polynomials have very high complexity~\cite{BrownD07}. When the sine function is involved, the SMT problem is undecidable, and only partial algorithms can be developed~\cite{DBLP:journals/lmcs/AvigadF06,DBLP:journals/jar/AkbarpourP10}.

Recently much attention has been given to developing practical solvers that incorporate scalable numerical computations. Examples of numerical algorithms that have been exploited include optimization algorithms~\cite{BorrallerasLNRR09,DBLP:conf/fmcad/NuzzoPSS10}, interval-based algorithms~\cite{HySAT,DBLP:conf/atva/EggersFH08,DBLP:conf/sefm/EggersRNF11,DBLP:conf/fmcad/Gao10}, Bernstein polynomials~\cite{bern}, and linearization algorithms~\cite{cordic}. These solvers have shown promising results on various nonlinear benchmarks in terms of scalability. 

However, for correctness-critical problems, there is always the concern that numerical errors can result in incorrect answers from numerically-driven solvers. For example, safety problems for hybrid systems can not be decided by numerical methods~\cite{andre07}. The problem is compounded by, for instance, the difficulty in understanding the effect of floating-point arithmetic in place of exact computation. There are two common ways of addressing these concerns. One is to use exact versions of the numerical algorithms, replacing floating-point operations by exact symbolic arithmetic~\cite{bern}; the other is to use post-processing (validation) procedures to ensure that only correct results are returned. Both options reduce the full power of numerical algorithms and are usually hard to implement as well. For instance, in the Flyspeck project~\cite{DBLP:conf/dagstuhl/Hales05} for the formal proof of the Kepler conjecture, validating the numerical procedures used in the original proof turns out to be the hardest computational part (and unfinished yet). In general, there has been no framework for understanding the actual performance guarantees of numerical algorithms in the context of decision problems. 

In this paper we aim to fill this gap by formally establishing the applicability of numerical algorithms in decision procedures, and the correctness guarantees they can actually provide. We do this as follows.

First, we introduce ``the $\delta$-SMT problem'' over real numbers, to capture what can in fact be {\em correctly} solved by numerically-driven procedures. Given an SMT formula $\varphi$, and any positive rational number $\delta$, the $\delta$-SMT problem asks for one of the following decisions:
%\vspace{-.1cm}
\begin{itemize}
\item $\mathsf{unsat}$: $\varphi$ is unsatisfiable.
\item $\delta$-$\mathsf{sat}$: The {\em $\delta$-weakening} of $\varphi$ is satisfiable. 
\end{itemize}
%\vspace{-.1cm}
Here, the $\delta$-weakening of $\varphi$ is defined as a numerical relaxation of the original formula. For instance, the $\delta$-weakening of $x=0$ is $|x|\leq\delta$. Note that if a formula is satisfiable, its $\delta$-weakening is always satisfiable. Thus, when a formula is $\delta$-{\sf sat}, either it is indeed satisfiable, or it is unsatisfiable but a $\delta$-perturbation on its numerical terms would make it satisfiable. The effect of this slight relaxation is striking. In sharp contrast to the undecidability of SMT for any signature extending real arithmetic by sine, we show that the bounded $\delta$-SMT problem for a wide range of nonlinear functions is decidable. 
%Here ``boundedness'' means that all the variables are considered within bounded intervals of reals. 
In fact, we show that the bounded $\delta$-SMT problem for the theory with exponentiation and trigonometric functions is $\mathsf{NP}$-complete, and $\mathsf{PSPACE}$-complete for theories with Lipschitz-continuous ODEs. We obtain these results using techniques from computable analysis~\cite{CAbook,vasco}. These results serve as the theoretical basis for our investigation of numerically-driven procedures. 

Next, if a decision algorithm can solve the $\delta$-SMT problem correctly, we say it is ``$\delta$-complete''. We propose to use $\delta$-completeness as the ideal correctness requirement on numerically-driven procedures, replacing the conventional notion of complete solvers (which can never be met in this context). This new notion makes it worthwhile to develop formally analyze numerical methods for decision problems and compare their strength, instead of viewing them as partial heuristics. As an example, we study DPLL$\langle$ICP$\rangle$, the integration of Interval Constraint Propagation (ICP)~\cite{newton} in DPLL(T)~\cite{DPbook}. It is a general solving framework for nonlinear formulas and has shown promising results~\cite{HySAT,DBLP:conf/fmcad/Gao10,DBLP:conf/sefm/EggersRNF11}. We obtain conditions that are sufficient and necessary for the $\delta$-completeness of DPLL$\langle$ICP$\rangle$. 

%We also show that correctness certificates of $\delta$-complete solvers can be produced at run-time and be externally checked, when it is hard to verify the numerical procedures used. 

Further, we show the applicability of $\delta$-complete procedures in correctness-critical practical problems. In bounded model checking~\cite{DBLP:journals/fmsd/ClarkeBRZ01,DBLP:books/daglib/0007403}, using a $\delta$-complete solver we return one of the following answers: either a system is absolutely safe up to some depth ($\mathsf{unsat}$ answers), or it would {\em become unsafe} under some $\delta$-bounded numerical perturbations ($\delta$-{\sf sat} answers). Since $\delta$ can be made very small, in the latter case the algorithm is essentially detecting robustness problems in the system: If a system would be unsafe under some small perturbations, it can hardly be regarded as safe in practice. Similar guarantees can be given for invariant validation and theorem proving. The conclusion is that, under suitable interpretations, the answers of numerically-driven decision procedures can indeed be relied on in correctness-critical applications, as long as they are $\delta$-complete. 
%%\vspace{-.35cm}

{\em Related Work.} Our goal is to provide a formal basis for the promising trend of numerically-driven decision procedures~\cite{BorrallerasLNRR09,DBLP:conf/fmcad/NuzzoPSS10,HySAT,DBLP:conf/atva/EggersFH08,DBLP:conf/sefm/EggersRNF11,DBLP:conf/fmcad/Gao10,bern,cordic}.
% give formally study the use of numerical computations in the context of decision procedures. for their reliable use in correctness-critical problems, such that the power of numerical algorithms can be fully exploited. This aims to serve as a basis for the promising-trend of numerically-driven solvers
Related attempt can be seen in Ratschan's work~\cite{DBLP:journals/jsc/Ratschan02}, in which he investigated the stability of first-order constraints under numerical perturbations. Our approach is, instead, to take numerical perturbations as a given and study its implications in practical applications. Results in this paper are related to our more theoretical results~\cite{gaoextended2} for arbitrarily-quantified sentences, where we do not analyze practical procedures. A preliminary notion of $\delta$-completeness was proposed by us earlier in~\cite{DBLP:conf/fmcad/Gao10} where only polynomials are considered.

The paper is organized as follows. In Section~\ref{smt} and \ref{delta} we define the bounded $\delta$-SMT problem and establish its decidability and complexity. In Section~\ref{dpllicp} we formally analyze DPLL$\langle$ICP$\rangle$ and discuss applications in Section~\ref{app}. 

\section{SMT with Type 2 Computable Functions}\label{smt}
\subsection{Basics of Computable Analysis}

Real numbers can be encoded as infinite strings, and a computability theory of real functions can be developed with oracle machines that perform operations using function-oracles encoding real numbers. This is the approach developed in Computable Analysis or Type 2 Computability~\cite{CAbook,Kobook,vasco}. We briefly review results of importance to us.

Throughout the paper we use $||\cdot||$ to denote $||\cdot||_{\infty}$ over $\mathbb{R}^n$ for various $n$. 

\begin{definition}[Names]
A name of $a\in \mathbb{R}$ is any function $\mathcal{\gamma}_a: \mathbb{N}\rightarrow \mathbb{Q}$ satisfying that $\forall i\in \mathbb{N}, |\gamma_a(i) - a|<2^{-i}.$ For $\vec a\in \mathbb{R}^n$, $\gamma_{\vec a}(i) = \langle \gamma_{a_1}(i), ..., \gamma_{a_n}(i)\rangle$.  
\end{definition}
Thus the name of a real number is a sequence of rational numbers converging to it. For $\vec a\in \mathbb{R}^n$, we write $\Gamma(\vec a) = \{\gamma: \gamma\mbox{ is a name of }\vec a\}$. 

A real function $f$ is computable if there is an oracle Turing machine that can take any argument $x$ of $f$ as a function oracle, and output the value of $f(x)$ up to an arbitrary precision. 
\begin{definition}[Computable Functions]
We say $f:\subseteq\mathbb{R}^n\rightarrow \mathbb{R}$ is computable if there exists a function-oracle Turing machine $\mathcal{M}_f$, outputting rational numbers, such that $\forall \vec x \in \dom(f)\ \forall \gamma_{\vec x}\in \Gamma(\vec x)\ \forall i \in \mathbb{N}, |M_f^{\gamma_{\vec x}}(i) - f(\vec x)|<2^{-i}.$
\end{definition}

In the definition, $i$ specifies the desired error bound on the output of $M_f$ with respect to $f(\vec x)$. For any $\vec x\in \dom(f)$, $M_f$ has access to an oracle encoding the name $\gamma_{\vec x}$ of $\vec x$, and output a $2^{-i}$-approximation of $f(\vec x)$. In other words, the sequence $M_f^{\gamma_{\vec x}}(1), M_f^{\gamma_{\vec x}}(2), ...$ is a name of $f(\vec x)$. Intuitively, $f$ is computable if an arbitrarily good approximation of $f(\vec x)$ can be obtained using any good enough approximation to any $\vec x\in\dom(f)$. A key property of this notion of computability is that computable functions over reals are continuous~\cite{CAbook}. Moreover, over any compact set $D\subseteq \mathbb{R}^n$, computable functions are uniformly continuous with a {\em computable modulus of continuity} defined as follows. 
\begin{definition}[Uniform Modulus of Continuity]
Let $f:\subseteq \mathbb{R}^n\rightarrow \mathbb{R}$ be a function and $D\subseteq\dom(f)$ a compact set. The function $m_f: \mathbb{N}\rightarrow \mathbb{N}$ is called a uniform modulus of continuity of $f$ on $D$ if\ $\forall \vec x,\vec y\in D$, $\forall i\in \mathbb{N}$, $||\vec x-\vec y||<2^{-m_f(i)}$ implies $|f(\vec x)-f(\vec y)|<2^{-i}.$
\end{definition}
\begin{proposition}[\cite{CAbook}]
Let $f:\subseteq\mathbb{R}^n\rightarrow \mathbb{R}$ be computable and $D\subseteq\dom(f)$ a compact set. Then $f$ has a computable uniform modulus of continuity over $D$.
\end{proposition}
Intuitively, if a function has a computable uniform modulus of continuity, then fixing any desired error bound $2^{-i}$ on the outputs, we can compute a {\em global} precision $2^{-m_f(i)}$ on the inputs from $D$ such that using any $2^{-m_f(i)}$-approximation of any $\vec x\in D$, $f(\vec x)$ can be computed within the error bound. 

Most common continuous real functions are computable~\cite{CAbook}. Addition, multiplication, absolute value, $\min$, $\max$, $\exp$, $\sin$ and solutions of Lipschitz-continuous ordinary differential equations are all computable functions. Compositions of computable functions are computable.

Moreover, complexity of real functions can be defined over compact domains. 
\begin{definition}[\cite{Ko92}]
Let $D\subseteq \mathbb{R}^n$ be compact. A real function $f:D\rightarrow\mathbb{R}$ is $\mathsf{P}$-computable ($\mathsf{PSPACE}$-computable), if it is computable by an oracle Turing machine $M_{f}^{\gamma(\vec x)}(i)$ that halts in polynomial-time (polynomial-space) for every $i\in \mathbb{N}$ and every $\vec x\in \dom(f)$. 
\end{definition}

We say $f$ is in Type 2 complexity class $\mathsf{C}$ if it is $\mathsf{C}$-computable. $f$ is $\mathsf{C}$-complete if it is $\mathsf{C}$-computable and $\mathsf{C}$-hard~\cite{Kobook}. If $f:D\rightarrow \mathbb{R}$ is $\mathsf{C}$-computable, then it has a $\mathsf{C}$-computable modulus of continuity over $D$. Polynomials, $\exp$, and $\sin$ are all $\mathsf{P}$-computable functions. A recent result~\cite{Kawamura09} established that the complexity of computing solutions of Lipschitz-continuous ODEs over compact domains is a {\sf PSPACE}-complete problem. 
%\vspace{-.2cm}
\subsection{Bounded SMT over $\mathbb{R}_{\mathcal{F}}$}

We now let $\mathcal{F}$ denote an arbitrary collection of Type 2 computable functions. $\mathcal{L}_{\mathcal{F}}$ denotes the first-order signature and $\mathbb{R}_{\mathcal{F}}$ is the standard structure $\langle \mathbb{R}, \mathcal{F}\rangle$. We can then consider the SMT problem over $\mathbb{R}_{\mathcal{F}}$, namely, satisfiability of quantifier-free $\mathcal{L}_{\mathcal{F}}$-formulas over $\mathbb{R}_{\mathcal{F}}$. We consider formulas whose variables take values from bounded intervals. Because of this, it is more convenient to directly write the bounds on existential quantifiers and express bounded SMT problems as $\Sigma_1$-sentences with bounded quantifiers.
\begin{definition}[Bounded $\Sigma_1$-Sentences]
A bounded $\Sigma_1$-sentence in $\mathcal{L}_{F}$ is
$$\varphi:\ \exists^{I_1}x_1\cdots \exists^{I_n}x_n. \psi(x_1,...,x_n).$$
\begin{itemize}
\item For all $i$, $I_i\subseteq \mathbb{R}$ is a bounded (open or closed) interval with rational endpoints. 
\item Each bounded quantifier $\exists^{I_i}x_i.\phi$ denotes $\exists x_i.(x_i\in I_i\wedge \phi)$. 
\item $\psi(x_1,...,x_n)$ is a quantifier-free $\mathcal{L}_{\mathcal{F}}$-formula, i.e., a Boolean combination of atomic formulas of the form $f(x_1,...,x_n)\circ 0$, where $f$ is a composition of functions in $\mathcal{F}$ and $\circ\in\{<,\leq, >, \geq, =, \neq \}$. 
\item We write $\dom(\varphi)= I_1\times \cdots \times I_n$, and require that all the functions occurring in $\psi(\vec x)$ are defined everywhere over its closure $\overline{\dom(\varphi)}$.
\end{itemize}
We can write a bounded $\Sigma_1$-sentence as $\exists^{\vec I}\vec x.\psi(\vec x)$ for short.
\end{definition}
%A first observation is that any bounded $\Sigma_1$-sentence can be put into the following standard form, where inequalities are implicitly expressed by the bounded quantifiers, and the atomic formulas only involve equalities. 
\begin{lemma}[Standard Form]\label{pre1}
Any bounded $\Sigma_1$-sentence $\varphi$ in $\mathcal{L}_{\mathcal{F}}$ is equivalent over $\mathbb{R}_{\mathcal{F}}$ to a sentence of the following form:
%\vspace{-.2cm}
$$\exists^{I_1}x_1\cdots \exists^{I_n}x_n\;\bigwedge_{i=1}^m(\bigvee_{j=1}^{k_i} f_{ij}(\vec x)=0).$$ 
\end{lemma}
%\vspace{-.6cm}
\begin{proof}
Assume that $\varphi$ is originally $\exists^{\vec I}\vec x\;\bigwedge_{i=1}^m (\bigvee_{j=1}^{k_i} g_{ij}(\vec x)\circ 0), \mbox{ where }\circ\in\{<,\leq, >, \geq, =,\neq\}.$ We apply the following transformations:

1. {\bf (Eliminate $\neq$)} Substitute each atomic formula of the form $g_{ij}\neq 0$ by $g_{ij}<0\vee g_{ij}>0$.

2. {\bf (Eliminate $\leq, <$)} Substitute $g_{ij}\leq 0$ by $-g_{ij}\geq 0$, and $g_{ij}<0$ by $-g_{ij}>0$. Now the formula is rewritten to $\exists^{\vec I}\vec x. \bigwedge_{i=1}^m (\bigvee_{j=1}^{k_i} g'_{ij}(\vec x)\circ 0), \mbox{ where } \circ\in\{>, \geq, = \}.$ ($g'_{ij} = -g_{ij}$ if the inequality is reversed; otherwise $g'_{ij}=g_{ij}$.)

3. {\bf (Eliminate $\geq, >$)} Substitute $g'_{ij}\geq 0$ (or $g'_{ij}>0$) by $g'_{ij}-v_{ij}=0,$ where $v_{ij}$ is a newly introduced variable, and add an innermost bounded existential quantifier $\exists v_{ij}\in I_{v_{ij}}$, where $I_{v_{ij}} = [0, m_{v_{ij}}]$ ($I_v=(0,m_{v_{ij}}]$). Here, $m_{v_{ij}}\in \mathbb{Q}$ is any value greater than the maximum of $g'_{ij}$ over $\overline{\dom(\varphi)}$. Note that such maximum of $g'_{ij}$ always exists over $\overline{\dom(\varphi)}$, since $g'_{ij}$ is continuous on $\overline{\dom(\varphi)}$, which is a compact, and is computable~\cite{Kobook}. 

The formula is now in the form $\exists^{\vec I}\vec x\exists^{\vec I_{\vec v}} \vec v.\ \bigwedge_{i=1}^m (\bigvee_{j=1}^{k_i} f_{ij}(\vec x, \vec v)=0),$ where $f_{ij} = g'_{ij}-v_{ij}$ if $v_{ij}$ has been introduced in the previous step; otherwise, $f_{ij} = g'_{ij}$. The new formula is in the standard form and equivalent to the original.
\qed\end{proof}
\begin{example}
A standard form of $\exists^{[-1,1]} x\exists^{[-1,1]} y\exists^{[-1,1]} z\;(e^z<x\rightarrow y<\sin(x))$ is $\exists^{[-1,1]} x\exists^{[-1,1]} y \exists^{[-1,1]} z \exists^{[0, 10]} u \exists^{(0,10]} v\;(e^z- x - u = 0) \vee (\sin(x)- y - v = 0).$
\end{example}
%\begin{remark}\label{max}
%Computing the exact maximum of $f$ over a bounded domain is hard. However, exactness is not needed in the transformations since $m_{v_{ij}}$ can be {\em any} value greater than the maximum and any loose bound on $f$ would work. In any case, to avoid ambiguity, our results will always be claimed for input formulas given in the standard form.
%\end{remark}

Recall that we allow the interval bounds on variables to be either open or closed. Let $\overline{S}$ and $S^o$ denote the closure and interior of any set $S$ over the reals. Based on our need we can consider the closure or the interior of the domains in a $\Sigma_1$-sentence. 
\begin{definition}[Closure and Interior]
Let $\varphi:= \exists^{I_1}x_1\cdots\exists^{I_n}x_n. \psi(\vec x)$ be a bounded $\Sigma_1$-sentence in $\mathcal{L}_{\mathcal{F}}$, we define the closure and interior of $\varphi$ as:
\begin{align*}
\overline{\varphi} &:= \exists^{\overline{I_1}} x_1\cdots\exists^{\overline{I_n}}x_n. \psi(\vec x) &\mathrm{(Closure)}\\
{\varphi}^o &:= \exists^{I_1^o} x_1\cdots\exists^{I_n^o}x_n. \psi(\vec x) &\mathrm{(Interior)}
\end{align*}
\end{definition}
\begin{proposition}
For any $\Sigma_1$-sentence $\varphi$, $\varphi^o\rightarrow \varphi$ and $\varphi\rightarrow \overline{\varphi}$.
\end{proposition}
\section{The Bounded $\delta$-SMT Problem}\label{delta}
The key for bridging numerical procedures and SMT problems is to introduce syntactic perturbations on $\Sigma_1$-sentences in $\mathcal{L}_{\mathcal{F}}$. 

\begin{definition}[$\delta$-Weakening and Perturbations]\label{weak-def}
Let $\delta\in \mathbb{Q}^+\cup\{0\}$ be a constant and $\varphi$ be a $\Sigma_1$-sentence in standard form:
%\vspace{-.3cm}
\[\varphi:= \exists^{\vec I}\vec x.\bigwedge_{i=1}^m (\bigvee_{j=1}^{k_i} f_{ij}(\vec x)= 0).
%\vspace{-.3cm}
\]
The $\delta$-weakening of $\varphi$ defined as:
%\vspace{-.3cm}
\[\varphi^{\delta}:= \exists^{\vec I} \vec x.\bigwedge_{i=1}^m(\bigvee_{j=1}^k |f_{ij}(\vec x)|\leq \delta).\]
Also, a $\delta$-perturbation is a constant vector $\vec c = (c_{11},...,c_{mk_m})$, $c_{ij}\in\mathbb{Q}$, satisfying $||\vec c||\leq\delta$, such that the $\vec c$-perturbed form of $\varphi$ is given by:
%\vspace{-.2cm}
\[\varphi^{\vec c}:= \exists^{\vec I} \vec x.\bigwedge_{i=1}^m(\bigvee_{j=1}^k f_{ij}(\vec x) = c_{ij}).\]
\end{definition}

\begin{proposition} 
$\varphi^{\delta}$ is true iff there exists a $\delta$-perturbation $\vec c$ such that $\varphi^{\vec c}$ is true. In particular, $\vec c$ can be the zero vector, and thus $\varphi\rightarrow\varphi^{\delta}$. 
\end{proposition}

We now define the bounded $\delta$-SMT problem. We follow the convention that SMT solvers return sat/unsat, which is equivalent to the corresponding $\Sigma_1$-sentence being true/false. 

%Following conventional terminology, we still regard an SMT problem as a quantifier-free formula $\varphi(\vec x)$, and explicitly write $\exists^{\vec I}\vec x.\varphi(\vec x)$ to denote the corresponding $\Sigma_1$-sentence. 

\begin{definition}[Bounded $\delta$-SMT in $\mathcal{L}_{\mathcal{F}}$] Let $\mathcal{F}$ be a finite collection of Type 2 computable functions. Let $\varphi$ be a bounded $\Sigma_1$-sentence in $\mathcal{L}_{\mathcal{F}}$ in standard form. The {\em bounded $\delta$-SMT problem} asks for one of the following two decisions on $\varphi$:
\begin{itemize}
\item $\mathsf{unsat}:$ $\varphi$ is false.
\item $\delta$-$\mathsf{sat}:$ $\varphi^{\delta}$ is true. 
\end{itemize}
When the two cases overlap, either decision can be returned. 
\end{definition}

Our main theoretical claim is that the bounded $\delta$-SMT problem is decidable for $\delta\in \mathbb{Q}^+$. This is essentially a special case of our more general results for arbitrarily-quantified $\mathcal{L}_{\mathcal{F}}$-sentences~\cite{gaoextended2}. However, different from~\cite{gaoextended2}, here we defined the standard forms of SMT problems to contain only equalities in the matrix, on which the original proof does not work directly. Also, in \cite{gaoextended2} we relied on results from computable analysis that are not needed here. We now give a direct proof for the decidability of $\delta$-SMT and analyze its complexity. 

\begin{theorem}[Decidability] Let $\mathcal{F}$ be a finite collection of Type 2 computable functions and $\delta\in \mathbb{Q}^+$. The bounded $\delta$-SMT problem in $\mathcal{L}_{\mathcal{F}}$ is decidable.  
\end{theorem}

\begin{proof}
We describe a decision procedure which, given any bounded $\Sigma_1$-sentence $\varphi$ in $\mathcal{L}_{\mathcal{F}}$ and $\delta\in \mathbb{Q}^+$, decides whether $\varphi$ is false or $\varphi^{\delta}$ is true. Assume that $\varphi$ is in the form of Definition~\ref{weak-def}. 

First, we need a uniform bound on all the variables so that a modulus of continuity for each function can be computed. Suppose each $x_i$ is bounded by $I_i$, whose closure is $\overline{I_i} = [l_i, u_i]$. We write %e transform the closure of the formula into 
\begin{eqnarray*}
\overline{\varphi} :=  \exists^{[0,1]} x_1 \cdots \exists^{[0,1]} x_n\; \bigwedge_{i=1}^m(\bigvee_{j=1}^{k_i} f_{ij}\big(l_1+(u_1-l_1)x_1,...,l_n+(u_n-l_n)x_n\big)=0).
\end{eqnarray*}
%The interior $\varphi^o$ of $\varphi$ is defined by replacing each $[0,1]$ by $(0,1)$. 
From now on, $g_{ij} = f_{ij}(l_1+(u_1-l_1)x_1,...,l_n+(u_n-l_n)x_n)$. After the transformation, we have $\overline{\dom(\varphi)} = [0,1]\times\cdots\times [0,1]$, on which each $g_{ij}$ is computable (it is a composition of the finitely many computable functions in $\mathcal{F}$) and has a computable modulus of continuity $m_{g_{ij}}$. We write $\psi(\vec x)$ to denote the matrix of $\varphi$ after the transformation. 

Choose $r\in \mathbb{N}$ such that $2^{-r}<\delta/4$. Then for each $g_{ij}$, we use $m_{g_{ij}}$ to obtain $e_{ij} = m_{g_{ij}}(r)$. Choose $e\in \mathbb{N}$ such that 
%\vspace{-.3cm}
\begin{eqnarray}\label{e-def}
e\geq\max(e_{11},...,e_{mk_m})
\end{eqnarray} and write $\varepsilon = 2^{-e}$. We then have
\begin{eqnarray}\label{first}
\forall \vec x, \vec y\in \overline{\dom(\varphi)}\  (||\vec x-\vec y||<\varepsilon \rightarrow |g_{ij}(\vec x)-g_{ij}(\vec y)|<{\delta}/{4}).
%\vspace{-.2cm}
\end{eqnarray}

We now consider a finite $\varepsilon$-net of $\overline{\dom(\varphi)}$, i.e., a finite $S_{\varepsilon} \subseteq \overline{\dom(\varphi)}$, satisfying
\begin{eqnarray}\label{two}
\forall \vec x\in \overline{\dom(\varphi)}\;\exists \vec a\in S_{\varepsilon}\ ||\vec x-\vec a||<\varepsilon. 
\end{eqnarray}
In fact, $S_{\varepsilon}$ can be explicitly defined as
\begin{eqnarray}\label{three}
S_{\varepsilon} = \{(a_1,...,a_n): a_i = k\cdot\varepsilon, \mbox{ where } k\in \mathbb{N}, 0\leq k\leq 2^e \}.
\end{eqnarray}
Next, we evaluate the matrix $\psi(\vec x)$ on each point in $S_{\varepsilon}$, as follows. Let $\vec a\in S_{\varepsilon}$ be arbitrary. For each $g_{ij}$ in $\psi$, we compute $g_{ij}(\vec a)$ up to an error bound of $\delta/8$, and write the result of the evaluation as $\overline{g_{ij}(\vec a)}^{\delta/8}$. Then $|g_{ij}(\vec a)- \overline{g_{ij}(\vec a)}^{\delta/8}|<{\delta}/{8}.$ Note $\overline{g_{ij}(\vec a)}^{\delta/8}$ is a rational number. We then define
\begin{eqnarray}
\widehat{\psi}(\vec x):= \bigwedge_{i=1}^m \bigvee_{j=1}^{k_i} |\overline{g_{ij}(\vec x)}^{\delta/8}|<\delta/2.
\end{eqnarray}
Then for each $\vec a$, evaluating $\widehat{\psi}(\vec a)$ only involves comparison of rational numbers and Boolean evaluation, and $\widehat{\psi}(\vec a)$ is either true or false. Now, by collecting the value of $\widehat{\psi}$ on every point in $S_{\varepsilon}$, we have the following two cases. 

$\bullet$ Case 1: For some $\vec a\in S_{\varepsilon}$, $\widehat{\psi}(\vec a)$ is true. We show that $\varphi^{\delta}$ is true. Note that
\begin{eqnarray*}
\widehat{\psi}(\vec a) \Rightarrow \bigwedge_{i=1}^m \bigvee_{j=1}^{k_i} |\overline{g_{ij}(\vec a)}^{\delta/8}|<\delta/2 
%&\Rightarrow& \bigwedge_{i=1}^m \bigvee_{j=1}^{k_i} |g_{ij}(\vec a)|<\delta/2+\delta/8\\
&\Rightarrow&\bigwedge_{i=1}^m \bigvee_{j=1}^{k_i} |g_{ij}(\vec a)|<\delta\cdot 5/8.
\end{eqnarray*}
We need to be careful about $\vec a$, since it is an element in $\overline{\dom(\varphi)}$, not $\dom(\varphi)$. If $\vec a\in \dom(\varphi)$, then $\varphi^{\delta}$ is true, witnessed by $\vec a$. Otherwise, $\vec a\in \partial(\dom(\varphi))$. Then by continuity of $g_{ij}$, there exists $\vec a'\in \dom(\varphi)$ such that $\bigwedge_{i=1}^m \bigvee_{j=1}^{k_i} |g_{ij}(\vec a')|<\delta$. (Just let a small enough ball around $\vec a$ intersect $\dom(\varphi)$ at $\vec a'$.) That means $\varphi^{\delta}$ is also true in this case, witnessed by $\vec a'$.

$\bullet$ Case 2: For every $\vec a\in S_{\varepsilon}$, $\widehat{\psi}(\vec a)$ is false. We show that $\varphi$ is false. Note that
\begin{eqnarray*}
\neg\widehat{\psi}(\vec a) \Rightarrow \bigvee_{i=1}^m \bigwedge_{j=1}^{k_i} |\overline{g_{ij}(\vec a)}^{\delta/8}|\geq\delta/2 
%&\Rightarrow& \bigvee_{i=1}^m \bigwedge_{j=1}^{k_i} |g_{ij}(\vec a)|\geq\delta/2-\delta/8\\
&\Rightarrow&\bigvee_{i=1}^m \bigwedge_{j=1}^{k_i} |g_{ij}(\vec a)|\geq \delta\cdot 3/8. 
\end{eqnarray*}
Now recall condition (\ref{first}) and (\ref{two}). For an arbitrary $\vec x\in \dom(\varphi)$, there exists $\vec a\in S_{\varepsilon}$ such that $|g_{ij}(\vec x)-g_{ij}(\vec a)|<\delta/4$ for every $g_{ij}$. Consequently, we have $|g_{ij}(\vec x)|\geq \delta\cdot 3/8- \delta/4 = \delta/8.$ Thus, $\forall \vec x\in \dom(\varphi), \bigvee_{i=1}^m\bigwedge_{j=1}^{k_i}\; |g_{ij}(\vec x)|>0.$ This means $\neg \varphi$ is true, and $\varphi$ is false. 

In all, the procedure that decides either that $\varphi^{\delta}$ is true, or that $\varphi$ is false. 
\qed\end{proof}

We now analyze the complexity of the $\delta$-SMT problem. The decision procedure given above essentially evaluates the formula on each sample point. Thus, given an oracle for evaluating the functions, we can construct a nondeterministic Turing machine that randomly picks the sample points and decides the formula.

Most of the functions we are interested in (exp, sin, ODEs) are in Type 2 complexity class $\mathsf{P}$ or $\mathsf{PSPACE}$. To prove interesting complexity results, a technical restriction is that we need to bound the number of function compositions in a formula, because otherwise evaluating nested polynomial-time functions can be exponential in the number of nesting. Formally we define: 
\begin{definition}[Uniformly Bounded $\Sigma_1$-class]
Let $\mathcal{F}$ be a finite set of Type 2 computable functions, and $S$ a class of bounded $\Sigma_1$-sentences in $\mathcal{L}_{\mathcal{F}}$. Let $l,u\in \mathbb{Q}$ satisfy $l\leq u$. We say $S$ is uniformly ($l,u,\mathcal{F}$)-bounded, if $\forall\varphi\in S$ of the form $\exists^{I_1}x_1\cdots\exists^{I_n}x_n \bigwedge_{i=1}^{m}\bigvee_{j=1}^{k_i}\;f_{ij}(\vec x)=0$,
\begin{itemize}
\item $\forall 1\leq i\leq n$, $I_i\subseteq [l,u]$. 
\item Each $f_{ij}(\vec x)$ is contained in $\mathcal{F}$. 
\end{itemize}
\end{definition}
\begin{proposition}[\cite{Kobook}]\label{polym}
Let $\mathsf{C}$ be a Type 2 complexity class contained in $\mathsf{PSPACE}$. Then given any compact domain $D$, a $\mathsf{C}$-computable function has a uniform modulus of continuity over $D$ given by a polynomial function.
\end{proposition}

We are now ready to prove the main complexity claim.

\begin{theorem}[Complexity]
Let $\mathcal{F}$ be a finite set of functions in Type 2 complexity class $\mathsf{C}$, $\mathsf{P}\subseteq\mathsf{C}\subseteq\mathsf{PSPACE}$. The $\delta$-SMT problem for uniformly bounded $\Sigma_1$-classes in $\mathcal{L}_{\mathcal{F}}$ is in $\mathsf{NP^C}$. 
\end{theorem}
\begin{proof}
We describe a nondeterministic Turing machine with a function oracle of complexity $\mathsf{C}$, that can decide in polynomial-time the $\delta$-SMT problem for a uniformly bounded class.

The function oracle $\theta$ we use behaves as follows. Given strings $s$, $t$, and $d$ on the query tape, $\theta(s,t,d)$ looks up the function $f_s\in \mathcal{F}$ encoded by $s$ and returns the value of $f_s({\vec x}_t)$ up to an error bound of $2^{-d}$, where ${\vec x}_t$ is a rational vector encoded by $t$ taken as the argument of $f_s$. Since all the functions in $\mathcal{F}$ are in complexity class $\mathsf{C}$, $\theta(s,t,d)$ is a $\mathsf{C}$-oracle. 

For any symbol $s$, we write $\len(s)$ to denote its bit-length. For an integer $i$, we know $\len(i)= O(\log(i))$. For a rational number $d$, which is the ratio of coprime integers $p$ and $q$, $\len(d)= O(\len(p)+\len(q))= O(\log(pq))$. For a function $f$, $len(f)$ is the length of its name. We write $O(\poly(n))$ to denote $\bigcup_kO(n^k)$. 

Let $\varphi$ be the input formula as in Definition~\ref{weak-def}, where each $f_{ij}\in \mathcal{F}$. Suppose $\varphi$ is in a uniformly $(l,u,\mathcal{F})$-bounded class. 

First, we observe that $e$, defined in (\ref{e-def}), can be obtained in time $O(\poly(\len(\varphi)+\len(\delta)))$, and $e = O(\poly(\len(\varphi)+\len(\delta)))$ (thus $\len(e) = O(\len(\varphi)+\len(\delta))$). This can be seen as follows. First, $2^{-r}<\delta$, we know $r= O(\log(\delta)) = O(\len(\delta))$. Then for each $f_{ij}$, we use its uniform modulus of continuity over $[l,u]$, given by a polynomial $m_{f_{ij}}$ (Proposition~\ref{polym}), and obtain $e_{ij}^f = m_{f_{ij}}(r)$, in time $O(\poly(\len(r)))$ and $e_{ij}^f = O(\poly(r))$. Then we compute $e_{ij}$ for the function $g_{ij}$ by scaling $e_{ij}^f$, using $e_{ij}  = \lceil -\log(2^{-e^f_{ij}}/\max_{1\leq i\leq n}\{u_i-l_i\})\rceil$. Thus $e_{ij} = O(e_{ij}^f+\log(\max_i(u_i-l_i))) = O(\poly(\len(\delta)+\len(\varphi)))$. Finally, let $e$ be the biggest $e_{ij}$. It is then clear that $e = O(\poly(\len(\varphi)+\len(\delta)))$, obtainable in polynomial time. 

Next, our procedure evaluates the matrix of the formula on each point $\vec a\in S_{\varepsilon}$. Note from ($\ref{three}$) that $S_{\varepsilon}$ is of size exponential in $e$. Here we exploit the nondeterminism of the machine by randomly picking $0\leq k\leq 2^e$ on each dimension. Note that since $\log(k)\leq e$, we have $\len(k) = O(e) = O(poly(\len(\varphi)+\len(\delta)))$. Let $\vec a = (a_1,...,a_n)$ be the randomly picked point in $S_{\varepsilon}$. Following the above estimate of $\len(k)$ and $\len(\varepsilon) = O(\log(2^{-e}))= O(e)$, we have $\len(\vec a) = O(\poly(\len(\varphi)+\len(\delta)))$. 

Now we evaluate $\widehat{\varphi}(\vec a)$. With access to the $\mathsf{C}$-oracle specified above, this can be done in polynomial-time, as follows. For each $g_{ij}(\vec a)$, we query the oracle with $\theta(f_{ij}, \vec a_{lu}, \delta/8)$, where $\vec a_{lu}$ is $\vec a$ scaled by $[l_i, u_i]$ on each dimension. This query uses $O(\poly(\len(\varphi)+\len(\delta)))$-space on the query tape. The oracle then return the value of $\overline{f_{ij}(\vec a_{lu})}^{\delta/8}= \overline{g_{ij}(\vec a)}^{\delta/8}$, and since $\mathsf{C}\subseteq \mathsf{PSPACE}$, $\len(\overline{g_{ij}(\vec a)}^{\delta/8})$ is polynomial in the input. Next we evaluate each atom by comparing these values obtained from the oracle with $\delta/2$. This uses time $O(\poly(\len(\varphi)+\len(\delta)))$. Finally, if $\widehat{\psi}(\vec x)$ is true, we return $\delta$-$\mathsf{sat}$. Thus the problem is decided in nondeterministic polynomial-time using access to the $\mathsf{C}$-oracle. We can conclude that the $\delta$-SMT problem for a uniformly bounded class is in $\mathsf{NP^C}$.\qed
\end{proof}
\begin{remark}
The restriction of a uniformly bounded class of formulas is a technical one. For a class of formulas of interest, we can always choose a rich enough $\mathcal{F}$ that contains the compositions we need, and a loose enough uniform bound on the variables. 
\end{remark}

We can now obtain a precise characterization of the complexity for $\delta$-SMT problems in signatures of interest. Recall that most common functions, such as polynomials, $\exp$, $\sin$, are all $\mathsf{P}$-computable and Lipschitz-continuous ODEs are $\mathsf{PSPACE}$-complete. 
\begin{corollary}
Let $\mathcal{F}$ be a finite set of $\mathsf{P}$-time computable real functions, such as $\{+, \times, \exp, \sin\}$. The uniformly-bounded $\delta$-SMT problem for $\mathcal{L}_{\mathcal{F}}$ is $\mathsf{NP}$-complete.
\end{corollary}
\begin{proof}
Since the functions in $\mathcal{F}$ are $\mathsf{P}$-time computable, the $\delta$-SMT problem is in $\mathsf{NP^P}= \mathsf{NP}$. We only need to encode Boolean satisfiability for hardness. We need to be careful that no negations can be used. For any propositional formula $\phi(p_1,...,p_n)$, substitute $p_i$ by $x_i<0$ and $\neg p_i$ by $x_i>1$, and add $(x_i=0\vee x_i=1)$ as a clause to the formula. Add the quantifiers $\exists^{[-1,2]}x_i$ for each $x_i$. Then for any $\delta<0.5$, $\phi$ is satisfiable iff the translation is $\delta$-true, and unsatisfiable iff the translation is false. Note that the cases do not overlap.\qed 
\end{proof}
\begin{corollary}
Let $\mathcal{F}$ be a finite set of Lipschitz-continuous ODEs over compact domains. Then the uniformly-bounded $\delta$-SMT problem in $\mathcal{L}_{\mathcal{F}}$ is in $\mathsf{PSPACE}$, and there exists $\mathcal{L}_{\mathcal{F}}$ such that it is $\mathsf{PSPACE}$-complete.
\end{corollary}
\begin{proof}
We have $\mathsf{NP^{PSPACE}} = \mathsf{PSPACE}$. Since some ODEs are {\sf PSPACE}-complete to solve~\cite{Kawamura09}, there exists $\mathcal{L}_{\mathcal{F}}$ for which $\delta$-SMT problem is $\mathsf{PSPACE}$-complete. \qed
\end{proof}
%\vspace{-.4cm}
\section{$\delta$-Completeness of the DPLL$\langle$ICP$\rangle$ Framework}\label{dpllicp}
%\vspace{-.1cm}
We now give a formal analysis of the integration of ICP and DPLL(T) for solving bounded $\delta$-SMT. Our goal is to establish sufficient and necessary conditions under which such an integration is $\delta$-complete. 
%\vspace{-.5cm}
\subsection{Interval Constraint Propagation}

The method of Interval Constraint Propagation (ICP)~\cite{handbookICP} finds solutions of real constraints using a ``branch-and-prune" method, combining interval arithmetic and constraint propagation. The idea is to use interval extensions of functions to ``prune'' out sets of points that are not in the solution set, and ``branch'' on intervals when such pruning can not be done, until a small enough box that may contain a solution is found. A high-level description of the decision version of ICP is given in Algorithm 1 and we give formal definitions as follows.

\begin{definition}[Floating-Point Intervals and Hulls]
Let $\mathbb{F}$ denote the finite set of all floating point numbers with symbols $-\infty$ and $+\infty$ under the conventional order $<$. Let
%\vspace{-.2cm}
$\mathbb{IF} = \{[a,b]\subseteq \mathbb{R}: a,b\in \mathbb{F}, a\leq b\}$ denote the set of closed real intervals with floating-point endpoints, and $\mathbb{BF} = \bigcup_{n=1}^{\infty}\mathbb{IF}^n$ the set of {\em boxes} with these intervals. Let $S\subseteq \mathbb{R}$ be any set of real numbers, the hull of $S$ is written as $\mathrm{Hull}(S) = \bigcap \{I\in \mathbb{IF}: S\subseteq I\}.$
\end{definition}

For $I= [a, b]\in \mathbb{IF}$, we write $|I| = |b-a|$ to denote its size. 

%Occasionally we will also use real intervals, writte as $\mathbb{IR} = \{[a,b]: a,b\in \mathbb{R}\}$. 

\begin{definition}[Interval Extension (cf. \cite{handbookICP})]
Let $f:\subseteq\mathbb{R}^n\rightarrow \mathbb{R}$ be a real function. An interval extension operator $\sharp(\cdot)$ maps $f$ to a function $\sharp f:\subseteq \mathbb{BF}\rightarrow \mathbb{IF}$, such that 
%\vspace{-.1cm}
$\forall B\in\mathbb{BF}\cap \dom(\sharp f), \{f(\vec x):\vec x\in B\}\subseteq \sharp f(B).$
\end{definition}
\begin{example}
The natural extension of $f = 2\cdot(x+y)\cdot z$ is given by $\sharp f = [2,2]\cdot(I_x+I_y)\cdot I_z$, where the interval operations are defined as $[a_1,b_1]+[a_2, b_2] = [a_1+a_2, b_1+b_2]$ and $[a_1,b_1]\cdot[a_2,b_2] = [\min(a_1a_2,a_1b_2,b_1a_2,b_1b_2), \max(a_1a_2,a_1b_2,b_1a_2,b_1b_2)]$. 
\end{example}

%\vspace{-.5cm}
\begin{algorithm}\label{algo1}
\SetKwData{Left}{left}\SetKwData{This}{this}\SetKwData{Up}{up}
\SetKwFunction{Union}{Union}\SetKwFunction{FindCompress}{FindCompress}
\SetKwInOut{Input}{input}\SetKwInOut{Output}{output}
\Input{Constraints $f_1(x_1,...,x_n)=0,...,f_m(x_1,...,x_n)=0$, initial box $B^0 = I^0_1\times \cdots \times I^0_n$, box stack $S=\emptyset$, and precision $\varepsilon\in \mathbb{Q}^+$.}
\Output{{\sf sat} or {\sf unsat}.}
\BlankLine
$S.\mathrm{push}(B_0)$\;
\While{$S\neq \emptyset$}{\label{while}
$B\leftarrow S.\mathrm{pop}()$ \;
\While{$\exists 1\leq i \leq m, B\neq \mathrm{Prune}(B,f_i)$}{ 
%\For{$j\leftarrow 1$ \KwTo $m$}{
	$B\leftarrow\mathrm{Prune}(B, f_i)$ \;
%	\If{$B=\emptyset$}{break\;}
}
\If{$B\neq \emptyset$}
{\eIf{$\exists 1\leq i\leq n, |I_i|\geq \varepsilon$}{$\{B_1,B_2\}\leftarrow \mathrm{Branch}(B, i)$\;$S.\mathrm{push}(\{B_1,B_2\})$\;}{return {\sf sat}\;}}
}
return {\sf unsat}\;
\caption{High-Level ICP$_{\varepsilon}$ (decision version of Branch-and-Prune)}
\end{algorithm}
%\vspace{-.5cm}
In Algorithm 1, Branch$(B,i)$ is an operator that returns two smaller boxes $B' = I_1\times\cdots\times I_i'\times\cdots\times I_n$ and $B''=I_1\times \cdots\times I_i''\times \cdots\times I_n$, where $I_i\subseteq I_i'\cup I_i''$. To ensure termination it is assumed that there exists some constant $0<c<1$ such that $c\cdot |I_i|\leq |I_i'|$ and $c\cdot |I_i|\leq |I_i''|$ for all $i$.

 %The pruning operation $\mathrm{Prune}(B,f_i)$ will be explained in detail below. 

The key component of the algorithm is the $\mathrm{Prune}(B, f)$ operation. A simple example of a pruning operation is as follows.
\begin{example}
Consider $x-y^2 = 0$ with initial intervals $x\in [1,2]$ and $y\in [0,4]$. Let $\sharp f(I_x,I_y) = I_x-I_y^2$ be the natural interval extension of the left hand side. Since we know $0\not\in \sharp f([1,2],[2, 4])$, we can contract the interval on $y$ from $[0,4]$ to $[0,2]$ in one pruning step.  
\end{example}
%\paragraph{Pruning Operators.} The key component of the ICP algorithm is the $\mathrm{Prune}(B, f)$ operation in Algorithm 1. A simple example of a pruning operation is as follows.
%\begin{example}
%Consider $x-y^2 = 0$ with initial intervals $x\in [1,2]$ and $y\in [0,4]$. Let $\sharp f(I_x,I_y) = I_x-I_y^2$ be the natural interval extension of the left handside. Since we know $0\not\in \sharp f([1,2],[2, 4])$, we can contract the interval on $y$ from $[0,4]$ to $[0,2]$ in one pruning step.  
%\end{example}
%\vspace{.1cm}
In principle, any operation that contracts the intervals on variables can be seen as pruning. However, for correctness we need several formal requirements on the pruning operator in ICP$_{\varepsilon}$. 
\begin{notation}
For any $f:\mathbb{R}^n\rightarrow \mathbb{R}$, we write $Z_f=\{\vec a\in \mathbb{R}^n: f(\vec a)=0\}$.
\end{notation}
\begin{definition}[Well-defined Pruning Operators]\label{well}
Let $\mathcal{F}$ be a collection of real functions, and $\sharp$ be an interval extension operator on $\mathcal{F}$. A {\em well-defined (equality) pruning operator} with respect to $\sharp$ is a partial function $\mathrm{Prune}_{\sharp} : \subseteq \mathbb{BF}\times \mathcal{F}\rightarrow \mathbb{BF}$, such that $\forall f\in \mathcal{F}$, $B,B'\in \mathbb{BF}$, 
\begin{itemize}
\item (W1) $\mathrm{Prune}_{\sharp}(B, f)\subseteq B$;
\item (W2) If $(\mathrm{Prune}_{\sharp}(B,f))\neq \emptyset$, then $0\in \sharp f(\mathrm{Prune}_{\sharp}(B,f))$. 
\item (W3) $B \cap Z_f \subseteq \mathrm{Prune}_{\sharp}(B, f)$;
\end{itemize}
\end{definition}
When $\sharp$ is clear, we simply write $\mathrm{Prune}$. It specifies the following requirements. (W1) requires contraction, so that the algorithm always makes progress: branching always decreases the size of boxes, and pruning never increases them. (W2) requires that the result of a pruning is always a reasonable box that may contain a zero. Otherwise $B$ should have been pruned out. (W3) ensures that the real solutions are never discarded in pruning (called ``completeness'' in~\cite{handbookICP}). 
%These conditions overlap with the usual requirements for designing pruning operators in constraint satisfaction literature, but are not the same. For instance, idempotence ($\mathrm{Prune}(\mathrm{Prune}(B,f))=\mathrm{Prune}(B,f)$) is not needed here. 
%\begin{definition}
We use $\mathrm{Prune}(B, f_1,...,f_m)$ to denote the iterative application of $\mathrm{Prune}(\cdot, f_i)$ on $B$ for all $1\leq i\leq m$, until a fixed-point is reached. (Line 4-6 in Algorithm 1.)
%\end{definition}
\begin{proposition}
For all $i$, $Prune(B,f_1,...,f_m)\subseteq Prune(B,f_i)$. 
\end{proposition}

%\begin{definition}[Branch-and-Prune Derivation]
%Let $B^0, ..., B^n\in \mathbb{BF}$ be boxes. We write $B_0,...,B_{n-1} \vdash_{bp} B_n$ if $B_n = \mathrm{Prune}(B_{n-1})$ (i.e., the last operation is pruning), and for every $0\leq i< n-1$, either $\mathrm{Prune}(B_i, f) = B_{i+1}$ or $\mathrm{Branch}(B_i) = B_{i+1}$.
%\end{definition}
It is clear from the description of Algorithm 1 that the following properties hold.
\begin{lemma}\label{key-lemma}
Algorithm 1 always terminates. If it returns $\mathsf{sat}$ then there exists nonempty boxes $B,B'\subseteq B_0$, such that $||B||<\varepsilon$ and $B=\mathrm{Prune}(B',f_1,...,f_m)$. If it returns $\mathsf{unsat}$ then $\forall\vec a\in B_0$, there exists $B\subseteq B_0$ such that $\vec a\in B$ and $\mathrm{Prune}(B,f_1,...,f_m)= \emptyset$. 
\end{lemma}

Now we prove the main theorem. 
\begin{theorem}[$\delta$-Completeness of ICP$_{\varepsilon}$]\label{main-theorem}
Let $\delta\in \mathbb{Q}^+$ be arbitrary. We can find an $\varepsilon\in \mathbb{Q}^+$ such that the $\mathrm{ICP}_{\varepsilon}$ algorithm is $\delta$-complete for conjunctive $\Sigma_1$-sentences in $\mathcal{L}_{\mathcal{F}}$ (where $\mathsf{sat}$ is interpreted as $\delta$-$\mathsf{sat}$) if and only if the pruning operator in ICP$_{\varepsilon}$ is well-defined. 
\end{theorem}
\begin{proof}
We consider an arbitrary bounded existential $\mathcal{L}_{\mathcal{F}}$-sentence containing only conjunctions, written as $\varphi: \exists^{\vec I}\vec x. \bigwedge_{i=1}^m f_i(\vec x) = 0$. Let $B_0 = \vec I$ be the initial bounding box. 

Since all the functions in $\varphi$ are computable over $B_0$, each $f_i$ has a uniform modulus of continuity over $B_0$, which we write as $m_{f_i}$. Choose any $k\in \mathbb{N}$ such that $2^{-k}<\delta$. Then for any $\varepsilon_i<m_{f_i}(k)$, we have 
%\vspace{-.3cm}
\begin{eqnarray}\label{lp}
\forall \vec x, \vec y\in B_0, ||\vec x-\vec y||<\varepsilon_i\rightarrow |f_i(\vec x)-f_i(\vec y)|<\delta.
\end{eqnarray}
We now fix $\varepsilon$ to be any positive rational number smaller than $\min(\varepsilon_1,...,\varepsilon_m)$.
 
By the previous lemma, we know ICP$_{\varepsilon}$ terminates and returns either {\sf sat} or {\sf unsat}. We now prove the two directions of the biconditional.

$\Leftarrow$: Suppose the pruning operator in ICP$_{\varepsilon}$ is well-defined.

Suppose ICP$_{\varepsilon}$ returns ``$\delta$-sat'', then by Lemma \ref{key-lemma}, there exist $B, B'\subseteq B_0$ such that $B = \mathrm{Prune}(B',f_1,...,f_m)$ and $||B'||<\varepsilon$. Then by the (W2), we know that $0\in \sharp f_i(B_{n})$ for every $f_i$. Now, by the definition of $\varepsilon$, we know from (\ref{lp}) that for every $i$, $\forall \vec a\in B, |f_i(\vec a)-0|<\delta.$ Namely, any $\vec a\in B$ is a witness for $\varphi^{\delta}: \exists^{\vec I} \vec x\ |f(\vec x)|<\delta$. Thus the $\delta$-weakening of $\varphi$ is true. 

Suppose ICP$_{\varepsilon}$ returns ``unsat''. Suppose $\varphi$ is in fact satisfiable. Then there is a point $\vec a\in B_0$ such that $\psi(\vec a)$ is true. However, following Lemma~\ref{key-lemma}, $\vec a\in B$ for some $B\subseteq B_0$ and $\mathrm{Prune}(B_0,f_1,...,f_m) = \emptyset$. However, this contradicts condition (W3) of the pruning operator. 

$\Rightarrow$: We only need to show that without any one of the three conditions in Definition~\ref{well}, we can define a pruning operator that fails $\delta$-completeness. 

Without (W1), we define a pruning operator that always outputs intervals bigger than $\varepsilon$ (such as the initial intervals). Then the procedure never terminates. Note that the other two conditions are trivially satisfied in this case (for any $f$ and $B_0$ satisfying $0\in \sharp f(B_0)$). Without (W2), consider the function $f(x)=x^2+1$ with $x\in [-1,1]$. We can define a pruning operator such that $\mathrm{Prune}([-1,1],f) = [1,1]$. This operator satisfies the other two conditions. However, the returned result $[1,1]$ fails $\delta$-completeness for any $\delta$ smaller than 2, since $f(1) = 2$. Without (W3), we simply prune any set to $\emptyset$ and always return {\sf unsat}. This violates $\delta$-completeness, which requires that if {\sf unsat} is returned the formula must be indeed unsatisfiable. The other two conditions are also satisfied in this case.  
\qed
\end{proof}

In practice, pruning operators are defined based on {\em consistency conditions} from constraint propagation techniques. Many pruning operators are used in practice~\cite{handbookICP}. Following Theorem~\ref{main-theorem}, we only need to prove their well-definedness to ensure $\delta$-completeness. For instance:
%For instance, a well-studied notion is called {\em arc-consistency}. 
%\begin{example}[Arc-Consistency]
%Let $f(x_1,...,x_n)=0$ be a constraint over $\mathbb{R}^n$. We say the sets $D_1,...,D_n\subseteq\mathbb{R}$ are arc-consistent with $f(\vec x)=0$, if for every $i$, we have $D_i = D_i \cap \{a_i\in \mathbb{R}: \forall j\in [1, i-1]\cup[i+1,n], \exists a_j\in D_j.\ f(a_1,...,a_n)=0\}.$
%\end{example}
%To compute the arc-consistent sets, a pruning operator can be defined by the righthand side of the fixed-point equation above. However, this operator is not always computable using intervals with floating-point numbers since it involves computing the solution set of $f$ exactly. To relax this condition, the following used in practice.
\begin{definition}[Box-consistent Pruning~\cite{newton}]
We say $\pi_B: \mathbb{BF}\times\mathcal{F}\rightarrow \mathbb{BF}$ is box-consistent, if for all $f\in \mathcal{F}$ and $B= I_1\times \cdots \times I_n \subseteq \dom(f)$, the $i$-th interval of $\pi_B(B,f)$ is $I_i\cap \mathrm{Hull}\big(\{a_i\in \mathbb{R}: 0\in \sharp f(I_1,...,\mathrm{Hull}(\{a_i\}),..., I_n\}\big).$
\end{definition}
%%\vspace{-.5cm}
\begin{proposition}
The Box-consistent Pruning operator is well-defined.
\end{proposition}
%\vspace{-.5cm}
\subsection{Handling ODEs}
In this section we expand our language to consider solutions of the initial value problems (IVP) of Lipschitz-continuous ODEs. Let $t_0, T\in \mathbb{R}$ and $g:\mathbb{R}^n\rightarrow \mathbb{R}$ be a Lipschitz-continuous function, i.e., for all $\vec x_1, \vec x_2\in\mathbb{R}^n$, $|g(\vec x_1)-g(\vec x_2)|\leq c||\vec x_1-\vec x_2||$ for some constant $c$. Let $t_0, T\in \mathbb{R}$ satisfy $t_0\leq T$ and $\vec y_0\in \mathbb{R}^n$. An IVP problem is given by 
%{\begin{eqnarray}\label{ivp}
$$\frac{d\vec y}{dt} = g(\vec y(t))\mbox{ and } \ \vec y(t_0) = \vec y_0, \mbox{ where }t\in [t_0, T].$$
%\end{eqnarray}}
where $\vec y: [t_0, T]\rightarrow \mathbb{R}^n$ is called the {\em solution} of the IVP. Consider $\vec y(t)$ as $(y_1(t),...,y_n(t))$, then each component $y_i: [t, T]\rightarrow \mathbb{R}$ is a Type 2 computable function, and can appear in some signature $\mathcal{F}$. In fact, we can also regard $\vec y_0$ as an argument of $y_i$ and write $y_i(t_0, \vec y_0)$. This does not change computability properties of $y_i$, since following the Picard-Lindel\"of representation $\vec y(t) = \int_{t_0}^t g(\vec y(s))ds + \vec y_0$, $y_i(t)$ is only linearly dependent on $\vec y_0$. 

In practice, with an ICP framework, we can exploit interval solvers for IVP problems~\cite{DBLP:journals/amc/NedialkovJC99}, for pruning intervals on variables that appear in constraints involving ODEs. This direction has received much recent attention~\cite{DBLP:conf/sefm/EggersRNF11,DBLP:conf/atva/EggersFH08,DBLP:conf/cp/GoldsztejnMEH10,DBLP:journals/sttt/IshiiUH11}. 
%We aim to extend our formal analysis for showing $\delta$-completeness of solvers in this domain.

%\begin{definition}%[Interval-Based ODE Solving]\label{ode-solve}
Consider the IVP problem defined above, with $\vec y_0$ contained in a box $B_{t_0}\subseteq \mathbb{R}^n$. Let $t_0\leq t_1\leq ...\leq t_m = T$ be a set of points in $[t_0, T]$. An interval-based ODE solver returns a set of boxes $B_{t_1},...,B_{t_m}$ such that 
%\begin{eqnarray*}
$$\forall i\in \{1,...,m\},\; \vec [y(t_i; B_{t_0})] = \{\vec y(t): t_0\leq t\leq t_i, \vec y_0\in B_{\vec y_0}\}\subseteq B_{t_i}.$$
%\end{eqnarray*}
%\end{definition}
%This is the guaranteed behavior of practical interval-based solvers such as~\cite{DBLP:journals/amc/NedialkovJC99}. 
%It is clear that this gives an interval extension of the solutions of the ODEs.
%\begin{proposition}
Now let $y_i: [t_0, T]\times B_0 \rightarrow \mathbb{R}$ be the $i$-th component of the solution $\vec y$ of an IVP problem. Then interval-based ODE solvers compute interval extensions of $y_i$. 
%\end{proposition}
Thus, pruning operators that respect the interval extension computed by interval ODE solvers can be defined. It can be concluded from Theorem~\ref{main-theorem} that ICP$_{\varepsilon}$ is $\delta$-complete for equalities involving ODEs, as long as the pruning operator is well-defined. A simplest strategy is just to prune out any set of points outside the interval extension:
\begin{proposition}[Simple ODE-Pruning]
Let $y_i(t,\vec y_0)$ be the $i$-th component function of an IVP problem. Suppose $\sharp y_i$ is computed by an interval ODE solver. Then the pruning operator $\mathrm{Prune}(I, y_i) = I\cap \sharp y_i(I_t, B_{\vec y_0})$ is well-defined. 
\end{proposition}
%\vspace{-.5cm}
\subsection{DPLL$\langle$ICP$\rangle$}
Now consider the integration of ICP into the framework of DPLL(T), so that the full $\delta$-SMT problem can be solved. Given a formula $\varphi$, a DPLL$\langle$ICP$\rangle$ solver uses SAT solvers to enumerate solutions to the Boolean abstraction $\varphi^B$ of the formula, and uses ICP$_{\varepsilon}$ to decide the satisfiability of conjunctions of atomic formulas. %(More detailed description is given in the Appendix.) 
DPLL$\langle$ICP$\rangle$ returns {\sf sat} when ICP$_{\varepsilon}$ returns {\sf sat} to some conjunction of theory atoms witnessing the satisfiability of $\varphi^B$, and returns {\sf unsat} when ICP$_{\varepsilon}$ returns {\sf unsat} on all the solutions to $\varphi^B$. Thus, it follows naturally that using a $\delta$-complete theory solver ICP$_{\varepsilon}$, DPLL$\langle$ICP$\rangle$ is also $\delta$-complete. 

\begin{corollary}[$\delta$-Completeness of DPLL$\langle$ICP$\rangle$]
Let $\mathcal{F}$ be a set of real functions. Then the pruning operators in ICP$_{\varepsilon}$ are well-defined for $\mathcal{F}$, if and only if, DPLL$\langle$ICP$\rangle$ using ICP$_{\varepsilon}$ is $\delta$-complete for bounded $\Sigma_1$-sentences in $\mathcal{L}_{\mathcal{F}}$.
\end{corollary}

\begin{proof}
Let $\varphi$ be a bounded SMT problem $\exists^{\vec I}\vec x \bigwedge_i\bigvee_j f_{ij}(\vec x)=0.$
Its Boolean abstraction $\varphi^B$ is given by $\bigwedge_i\bigvee_j p_{ij}$, where $p_{ij}$ is the propositional abstraction of $f_{ij}(\vec x)=0$. 

Choose $\varepsilon$ to satisfy that $\forall \vec x, \vec y\in \vec I |f_{ij}(\vec x)-f_{ij}(\vec y)|<\delta$ for all $f_{ij}$ that appear in the $\varphi$.

Now, in the DPLL(T) framework, the SAT solver returns an assignment to $p_{ij}$ such that $\varphi^B$ evaluates to true, then ICP$_{\varepsilon}$ is used for checking the satisfiability of the corresponding conjunction of theory atoms. It is important to note that $\varphi^B$ does not contain negations. 

Suppose the pruning operator in ICP$_{\varepsilon}$ is well-defined. Then ICP$_{\varepsilon}$ is $\delta$-complete. Now, suppose DPLL$\langle$ICP$\rangle$ returns {\sf sat}. Then $\varphi^B$ is true witnessed by a set $\{p_1,...,p_m\}$ assigned to true, which in turn corresponds to a set $\{f_1(\vec x)=0,...,f_m(\vec x)=0\}$ of the theory atoms. By $\delta$-completeness of ICP$_{\varepsilon}$, we know that $\varphi^{\delta}$ is true. On the other hand, suppose $\varphi$ is decided as {\sf unsat}. Then either there is no assignment such that $\varphi^B$ is true, or for each satisfying assigment to $\varphi^B$, ICP$_{\varepsilon}$ decides that the corresponding set of theory atoms is not satisfiable. By $\delta$-completeness of ICP$_{\varepsilon}$, the {\sf unsat} answers are always correct. In all, DPLL$\langle$ICP$\rangle$ is also $\delta$-complete. 

Suppose the pruning operator in ICP$_{\varepsilon}$ is not well-defined, then DPLL$\langle$ICP$\rangle$ is simply not $\delta$-complete for conjunctions of theory atoms, and thus not $\delta$-complete for bounded SMT in $\mathcal{L}_{\mathcal{F}}$. 
\qed\end{proof}

\section{Applications}\label{app}
%\vspace{-.2cm}
%\subsection{Applications of $\delta$-Complete Procedures}
$\delta$-Complete solvers return answers that allow one-sided, $\delta$-bounded errors. The framework allows us to easily understand the implications of such errors in practical problems. Indeed, $\delta$-complete solvers can be {\em directly} used in the following correctness-critical problems. 
%\vspace{-.3cm}
\paragraph{Bounded Model Checking and Invariant Validation.} Let $S=\langle X, \mathsf{Init}, \mathsf{Trans}\rangle$ be a transition system over $X$, which can by continuous or hybrid. Then given a subset $U\subseteq X$, the bounded model checking problem asks whether $\varphi_n:= \exists \vec x_0,...,\vec x_n(\mathsf{\vec x_0}\wedge \bigwedge_{i=0}^{n-1} \mathsf{Trans}(\vec x_i, \vec x_{i+1}) \wedge \vec x_n\in U)$ is true. Here $U$ denotes the ``unsafe'' values of the system, and we say $S$ is safe up to $n$ if $\varphi_n$ is false. Thus, using a $\delta$-complete solver for $\varphi_n$, we can determine the following: If $\varphi_n$ is {\sf unsat}, then $S$ is indeed safe up to $n$; on the other hand, if $\varphi_n$ is {\sf $\delta$-sat}, then either the system is unsafe, or it would be unsafe under a $\delta$-perturbation, and a counterexample is provided by the certificate for {\sf $\delta$-sat}. This $\delta$ can be set by the user based on the intended tolerance of errors of the system. Thus, a $\delta$-complete solver can be directly used. 

For invariant validation, a proposed invariant $\mathsf{Inv}$ can prove safety if the sentence $\varphi:=\forall\vec x,\vec x'((\mathsf{Init}(\vec x)\rightarrow \mathsf{Inv}(\vec x))\wedge (\mathsf{Inv}(\vec x)\wedge\mathsf{Trans}(\vec x, \vec x')\rightarrow\mathsf{Inv}(\vec x'))\wedge \mathsf{Inv}(\vec x)\rightarrow \neg(U(\vec x)))$ is true. We then use a $\delta$-complete solver on $\neg\varphi$, which is existential. When {\sf unsat} is returned, $\mathsf{Inv}$ is indeed an inductive invariant proving safety. When {\sf $\delta$-sat} is returned, either $\mathsf{Inv}$ is not an inductive invariant, or under a small numerical perturbation, $\mathsf{Inv}$ would violate the inductive conditions.
%\vspace{-.4cm}
\paragraph{Theorem Proving.} For theorem proving, one-sided errors are not directly useful since no robustness problem is involved. We can still approach a statement $\varphi$ by making $\delta$-decisions on $\neg\varphi$, and refine $\delta$ when needed. Starting from any $\delta$, whenever {\sf unsat} is returned, $\varphi$ is proved; when {\sf $\delta$-sat}, we can try a smaller $\delta$. This reflects the common practice in proving these statements. 
%Also, proofs can be extracted from certified solvers when $\varphi$ is proved in this manner. 
%\vspace{-.8cm}

%\subsection{Implementation}
%\vspace{-.1cm}
%In this paper we focus on the rigorous development of the framework of $\delta$-completeness. As mentioned above, there now exists several numerically-driven solvers~\cite{BorrallerasLNRR09,DBLP:conf/fmcad/NuzzoPSS10,HySAT,DBLP:conf/atva/EggersFH08,DBLP:conf/sefm/EggersRNF11,DBLP:conf/fmcad/Gao10,bern,cordic} with promising results on various nonlinear problems. What we have proposed is a rationale for ensuring their applicability in correctness-critical problems. 
%We are developing our tool {\sf dReal}, implementing the DPLL$\langle$ICP$\rangle$ framework for solving problems in the signatures discussed in the paper. The emphasis is on the $\delta$-completeness of the solver: We provide certificates for {\sf $\delta$-sat} and {\sf unsat} answers as described above. Our tool is under active development, and the current prototype version, benchmarks, and experimental data are available at~\cite{dReal}. We have made use of many existing packages, including MiniSAT~\cite{een03minisat}, OpenSMT~\cite{DBLP:conf/tacas/BruttomessoPST10}, RealPaver~\cite{DBLP:journals/toms/GranvilliersB06}, and VNODE~\cite{DBLP:journals/amc/NedialkovJC99}. 
%\vspace{-.2cm}
\section{Conclusion}
%\vspace{-.1cm}
We introduced the notion of ``$\delta$-complete decision procedures'' for solving SMT problems over real numbers. Our aim is to provide a general framework for solving a wide range of nonlinear functions including transcendental functions and solutions of Lipschitz-continuous ODEs. $\delta$-Completeness serves as a replacement of the conventional completeness requirement on exact solvers, which is impossible to satisfy in this domain. We proved the existence of $\delta$-complete decision procedures for bounded SMT over reals with Type 2 computable functions and showed the complexity of the problem. We use $\delta$-completeness as the standard correctness requirement on numerically-driven decision procedures, and formally analyzed the solving framework DPLL$\langle$ICP$\rangle$. We proved sufficient and necessary conditions for its $\delta$-completeness. We believe our results serve as a foundation for the development of scalable numerically-driven decision procedures and their application in formal verification and theorem proving.

\section*{Acknowledgement}
We are grateful for many important suggestions from Lenore Blum and the anonymous reviewers. 

%\vspace{-.5cm}
\bibliographystyle{abbrv}
\bibliography{tau}
\end{document}